\documentclass{amsart}
\usepackage{amsfonts}
\usepackage{amsmath}
\usepackage{amssymb}
\usepackage{amsthm}
\usepackage{graphicx}
\usepackage{capt-of}
\usepackage{tikz}
\usepackage{hyperref}

\newtheorem{theorem}{Theorem}

\newtheorem{lemma}{Lemma}

\newcommand{\hol}[3][1]{\chi_{#1} \left(#2, #3\right)}
\newcommand{\G}[2][H]{\mathcal G_{#2} \left(#1\right)}
\newcommand{\PH}[2][H]{\mathcal P_{#2} \left(#1\right)}

\begin{document}

\title{A non-injective version of Wigner's theorem}
\author{Mark Pankov, Lucijan Plevnik}
\subjclass[2000]{}
\keywords{projection, self-adjoint operator of finite rank, Wigner's theorem}
\address{University of Warmia and Mazury, Faculty of Mathematics and Computer Science,  S{\l}oneczna 54, Olsztyn, Poland}
\address{University of Ljubljana, Faculty for Mathematics and Physics, Jadranska 21, Ljubljana, Slovenia}
\email{pankov@matman.uwm.edu.pl, lucijan.plevnik@fmf.uni-lj.si}
\maketitle

\begin{abstract}
Let $H$ be a complex Hilbert space and 
let ${\mathcal F}_{s}(H)$ be the real vector space of all self-adjoint finite rank operators on $H$.
We prove the following non-injective version of Wigner's theorem: 
every linear operator on ${\mathcal F}_{s}(H)$ sending rank one projections to rank one projections
(without any additional assumption) is either induced by a linear or conjugate-linear isometry or 
constant on the set of rank one projections.
\end{abstract}

\section{Introduction}
Wigner's theorem plays an important role in mathematical foundations of quantum mechanics.
Pure states of a quantum mechanical system are identified with rank one projections (see, for example, \cite{Var})
and Wigner's theorem \cite{Wigner} characterizes all symmetries of the space of pure states as unitary and anti-unitary operators.
We present a non-injective version of this result in terms of linear operators on 
the real vector space of self-adjoint finite rank operators which send rank one projections to rank  one projections.

Let $H$ be a complex Hilbert space.
For every natural $k<\dim H$ we denote by ${\mathcal P}_{k}(H)$ the set of all rank $k$ projections,
i.e. bounded self-adjoint idempotent operators of rank $k$.
Let ${\mathcal F}_{s}(H)$ be the real vector space of all self-adjoint finite rank operators on $H$.
This vector space is spanned by ${\mathcal P}_{k}(H)$, see e.g. \cite[Lemma 2.1.5]{Molnar-book}.

Classical Wigner's theorem says that 
every bijective transformation of ${\mathcal P}_{1}(H)$ preserving the angle between the images of any two projections,
or equivalently, preserving the trace of the composition of any two projections, is induced by a unitary or anti-unitary operator. 
The first rigorous proof of this statement was given in \cite{LM}, see also \cite{Uh} for the case when $\dim H\ge 3$.
By the non-bijective version of this result \cite{Bar, Barv, Geher1}, 
arbitrary (not necessarily bijective) transformation of ${\mathcal P}_{1}(H)$ preserving the angles between the images of projections
(it is clear that such a transformation is injective) is induced by a linear or conjugate-linear isometry.

Various analogues of Wigner's theorem for ${\mathcal P}_{k}(H)$ can be found in 
\cite{Geher,GeherSemrl,Gyory,Molnar1,Molnar-book,Molnar2,Pankov1,Pankov2,Pankov-book-w,Semrl}.
In particular, transformations of ${\mathcal P}_{k}(H)$ preserving principal angles between the images of any two projections 
and transformations preserving the trace of the composition of  any two projections 
are determined  in \cite{Molnar1,Molnar2} and \cite{Geher}, respectively.
All such transformations are induced by linear or conjugate-linear isometries, 
except in the case $\dim H=2k\ge 4$ when there is an additional class of transformations. 
The description of transformations preserving the trace of the composition given in \cite{Geher} is based on the following fact from \cite{Molnar1}:
every transformation of ${\mathcal P}_{k}(H)$ preserving the trace of the composition of  two projections can be extended to 
an injective linear operator on ${\mathcal F}_{s}(H)$.
So, there is an intimate relation between Wigner's type theorems  mentioned above and 
results concerning linear operators sending projections to projections \cite{ACh, Pankov3,SChM,Stormer}.

Consider a linear operator $L$ on ${\mathcal F}_{s}(H)$ such that
\begin{equation}\label{eqk}
L({\mathcal P}_{k}(H))\subset {\mathcal P}_{k}(H)
\end{equation}
such that the restriction of $L$ to ${\mathcal P}_{k}(H)$ is injective.
We also assume that $\dim H\ge 3$.
By \cite{Pankov3}, this operator is induced by a linear or conjugate-linear isometry if $\dim H\ne 2k$. 
In the case when $\dim H=2k$, it can be also a composition of an operator induced by a linear or conjugate-linear isometry and 
an operator which sends any projection on a $k$-dimensional subspace $X$ to the projection on the orthogonal complement $X^{\perp}$.
This statement is a small generalization of the result obtained in \cite{ACh}.
The main result of \cite{Pankov3} concerns linear operators sending ${\mathcal P}_{k}(H)$ to ${\mathcal P}_{m}(H)$,
as above, whose restrictions to ${\mathcal P}_{k}(H)$ are injective.

In this paper, we determine all possibilities for a linear operator $L$ on ${\mathcal F}_{s}(H)$ satisfying \eqref{eqk}
for $k=1$ without any additional assumption.
Such an operator is either induced by a linear or conjugate-linear isometry or its restriction to ${\mathcal P}_{1}(H)$ is constant. We mention that this result could be easily obtained from \cite[Theorem 2.1]{Semrl-adj}, as such a map $L$ clearly preserves the adjacency relation on the set ${\mathcal F}_{s}(H)$. However, we will present an elementary approach by only using the Wigner's theorem.

Some remarks concerning the case when $k>1$ will be given in the last section.

\section{The main result}
We investigate linear maps on ${\mathcal F}_{s}(H)$ preserving the set of projections of rank one. Our  main result is the following.

 \begin{theorem}\label{main}
 	Let $H$ be a complex Hilbert space, $\dim H \ge 2$, and $L \colon {\mathcal F}_{s}(H) \to {\mathcal F}_{s}(H)$ a linear map. Then we have
 	\begin{equation}\label{lpro}
 			L \left(\PH{1}\right) \subset \PH{1}
 		\end{equation}
 	if and only if either there exists $P_0 \in \PH{1}$ such that
 	\begin{eqnarray*}
 			L (A) = ({\rm tr} A) P_0, & A \in {\mathcal F}_{s}(H)
 		\end{eqnarray*}
 	or there exists a linear or conjugate-linear isometry $U \colon H \to H$ such that
 	\begin{eqnarray*}
 			L (A) = UAU^\ast, & A \in {\mathcal F}_{s}(H).
 		\end{eqnarray*}
 \end{theorem}
 
\section{Preliminaries}
Denote by $P_{X}$ the projection whose image is  a closed subspace $X\subset H$.
Since $P_{X}$ belongs to ${\mathcal P}_{k}(H)$ if and only if $X$ is $k$-dimensional,
${\mathcal P}_{k}(H)$ will be identified with the Grassmannian ${\mathcal G}_{k}(H)$. For any subspace $Z \subset H$, denote
$$
\left\langle Z \right]_1 = \{ X \in \G 1 \mid X \subset Z \}.
$$
If $\dim H \ge 2$, then ${\mathcal G}_{1}(H)$ is a projective space, whose projective lines are exactly sets of the form $\left\langle S \right]_1$, $S\in{\mathcal G}_{2}(H)$.

We will show that the maps $f$, satisfying \eqref{lpro}, behave nicely on projective lines in ${\mathcal G}_{1}(H)$. In order to do that, we will need the following concept, which is a modification of the concept, introduced in \cite{Geher}. For any $X, Y \in \G 1$ and $t \in \left(\tfrac{1}{2}, \infty\right)$, define the set
$$
\hol[t]{X}{Y} = \left\{ Z \in \G 1 \, : \, t\left(P_X + P_Y\right) + (1-2t) P_Z \in \PH 1 \right\}.
$$
The following lemma describes this set.

\begin{lemma}\label{hole}
Let $X, Y \in {\mathcal G}_{1}(H)$ and $t \in \left(\tfrac{1}{2}, \infty\right)$. Then
the following statements hold.
\begin{itemize}
	\item $\hol[t]{X}{Y}\subset \left\langle X+Y \right]_1$
	\item $\hol[t]{X}{Y} \ne \emptyset \iff {\rm tr}\, (P_XP_Y) \ge \left(1-\tfrac{1}{t}\right)^2$
	\item If $X$ and $Y$ are orthogonal, then $\hol[1]{X}{Y} = \left\langle X+Y\right]_1$.
	\item If $X \ne Y$ and ${\rm tr}\, (P_XP_Y) > \left(1-\tfrac{1}{t}\right)^2$, then $\hol[t]{X}{Y}$
	is homeomorphic to a circle.
	\item If $X = Y$ or ${\rm tr}\, (P_XP_Y) = \left(1-\tfrac{1}{t}\right)^2 \ne 0$, then $\hol[t]{X}{Y}$ is a singleton.
\end{itemize}
\end{lemma}

\begin{proof}
	It is easy to show that $\hol[t]{X}{X} = \{X\}$.
	
Assume now that $X \ne Y$ and denote $S = X+Y$ and $A = P_X+P_Y$.
Then $A$ is a positive semidefinite operator with trace $2$. Its kernel equals $X^\bot \cap Y^\bot$, so its range equals $S$. Therefore, if $Z \in \hol[t]{X}{Y}$, then $t A + (1-2t) P_Z$ is positive semidefinite, implying that $Z \in \left\langle S \right]_1$.

Moreover, there exist $c \in [0,1)$ and an orthonormal base $\mathcal B$ of $S$, according to which we have the matrix representation
$$
A|_S = \left[\begin{array}{cc}
	1 + c & 0 \\
	0 & 1 - c
\end{array}\right].
$$
Note that
$$
{\rm tr}\, (P_XP_Y) = \tfrac{1}{2}{\rm tr}\, (A^2-A) = c^2.
$$
If $Z$ is any element of $\left.\langle S\right]_1$, then, according to $\mathcal B$,
\begin{eqnarray*}
	P_Z|_S = \left[\begin{array}{cc}
		s & w \\
		\overline w & 1-s
	\end{array}\right]
\end{eqnarray*}
for some $s \in [0,1]$ and $w \in \mathbb C$, $|w| = \sqrt{s(1-s)}$. Any such $Z$ belongs to $\hol[t]{X}{Y}$ if and only if
$$
\det \left(t \left[\begin{array}{cc}
	1 + c & 0 \\
	0 & 1 - c
\end{array}\right] + (1-2t)\left[\begin{array}{cc}
s & w \\
\overline w & 1-s
\end{array}\right]\right) = 0.
$$
A straightforward calculation shows that the latter holds if and only if we have either $c = 0$ and $t = 1$ or $c \ne 0$ and $s$ equals
\begin{equation}\label{ct}
	\tfrac{(1 + c) (t(1 + c)-1)}{2c (2t-1)}.
\end{equation}
Thus, if $X$ and $Y$ are orthogonal, then $\hol[t]{X}{Y}$ is non-empty if and only if $t = 1$ and in this case, it equals $\left\langle X+Y\right]_1$. In the case when they are not orthogonal, $\hol[t]{X}{Y}$ is non-empty if and only if \eqref{ct} belongs to $[0,1]$, which is equivalent to $c \ge \left|1-\tfrac{1}{t}\right|$. Next, if \eqref{ct} belongs to $\{0,1\}$, which is equivalent to $c = \left|1-\tfrac{1}{t}\right|$, then $\hol[t]{X}{Y}$ is a singleton. Finally, if \eqref{ct} belongs to $(0,1)$ and equals $s$, then any $Z \in \hol[t]{X}{Y}$ can be identified with an element $w$ of the circle with origin $0$ and radius $\sqrt{s(1-s)}$.
\end{proof}

\section{Proof of Theorem 1}
Recall that $L$ is a linear map ${\mathcal F}_{s}(H) \to {\mathcal F}_{s}(H)$ satisfying \eqref{lpro}. Denote by $f$ the transformation ${\mathcal G}_{1}(H) \to {\mathcal G}_{1}(H)$, induced by $L$, i.e.
$L(P_{X})=P_{f(X)}$, $X\in {\mathcal G}_{1}(H)$.

\begin{lemma}\label{lemma1}
The following assertions are fulfilled:
\begin{enumerate}
\item For any $t \in \mathbb R \setminus \left\{0, \tfrac{1}{2}\right\}$ and $X,Y\in {\mathcal G}_{1}(H)$ we have $$f(\hol[t]{X}{Y})\subset \hol[t]{f(X)}{f(Y)}.$$
If $f(X)=f(Y)$, then $f$ is constant on $\hol[t]{X}{Y}$.
\item $f$ transfers any projective line to a subset of a projective line.
\end{enumerate}
\end{lemma}

\begin{proof}
	\begin{enumerate}
		\item Easy verification.
		\item If $S\in{\mathcal G}_{2}(H)$ and $X,Y\in \left\langle S \right]_1$ are orthogonal, 
		then $\hol{X}{Y}$ coincides with $\left\langle S \right]_1$ by Lemma \ref{hole} and we have 
		$$f(\left\langle S \right]_1)\subset \hol{f(X)}{f(Y)}\subset \left\langle S' \right]_1$$
		with $S'=f(X)+f(Y)$ if $f(X)\ne f(Y)$, otherwise we take any $2$-dimensional subspace $S'$ containing $f(X)=f(Y)$.
	\end{enumerate}
\end{proof}

\begin{lemma}\label{const}
The restriction of $f$ to any projective line is either injective or constant.
\end{lemma}

\begin{proof}
Let $S\in{\mathcal G}_{2}(H)$.
Suppose that the restriction of $f$ to $\left\langle S \right]_1$ is not injective. Then there exist distinct $X,Y\in \left\langle S \right]_1$ such that $f(X) = f(Y)$. For every $t \in \mathbb R$ define
$$
g(t) = \det ((t(P_X+P_Y)+(1-2t)P_{S \cap X^\bot})|_S).
$$
Then $g \colon \mathbb R \to \mathbb R$ is a continuous function such that $g(\tfrac{1}{2}) > 0$. Let $A = P_X+P_Y-P_{S \cap X^\bot}$. Let $\langle \cdot, \cdot \rangle$ denote the scalar product on $H$. Since $\langle Ax,x \rangle > 0 $ for $x \in X$ and $\langle Az,z \rangle \le 0 $ for $z \in S \cap X^\bot$, we have $g(1) \le 0$. Therefore, there exists $t \in \left(\tfrac{1}{2},1\right]$ such that $g(t) = 0$. For such $t$ we have $S \cap X^\bot \in \hol[t]{X}{Y}$. By Lemma \ref{lemma1}, $f(S \cap X^\bot) = f(X) = f(Y)$. Another application of Lemma \ref{lemma1} yields that $f$ is constant on $\hol{X}{S \cap X^\bot}$, which equals $\left\langle S \right]_1$ by Lemma \ref{hole}.
\end{proof}

\begin{lemma}\label{contin}
	The restriction of $f$ to any projective line is continuous.
\end{lemma}

\begin{proof}
	Let $S\in{\mathcal G}_{2}(H)$. Then the linear span of $\{P_X \mid X \in \left\langle S \right]_1\}$ is finite-dimensional, which implies that the restriction of $L$ to this linear span is bounded. Hence,	the restriction of $f$ to $\left\langle S \right]_1$ is continuous.
\end{proof}

\begin{proof}[Proof of Theorem \ref{main}]
	The two examples in the conclusion of the theorem clearly satisfy \eqref{lpro}. Assume now that \eqref{lpro} holds.
	
	If $f$ is constant, then $\phi(A) = ({\rm tr} A) P_0$, $A \in \mathcal F_s \left(H\right)$, for some $P_0 \in \PH{1}$.
	
	Assume now that $f$ is not constant. Then there exist $X, Y \in \G{1}$ such that $f(X) \ne f(Y)$. Denote $S = X+Y \in \G{2}$. We will first show that
	\begin{equation}\label{gc}
		f (\left\langle S \right]_1) = \left\langle f(X)+f(Y) \right]_1.
	\end{equation}
	By Lemma \ref{lemma1}, Lemma \ref{const}, and Lemma \ref{contin}, $f$ is an injective continuous map from $\left\langle S \right]_1$ to $\left\langle f(X)+f(Y) \right]_1$, which are both homeomorphic to the $2$-dimensional sphere $\mathbb S^2$. Thus, $f$ induces an injective continuous map $\widetilde f \colon \mathbb S^2 \to \mathbb S^2$. If $\widetilde{f}$ was not surjective, then it would map into $\mathbb S^2 \setminus \{p\}$ for some $p \in \mathbb S^2$, which is homeomorphic to $\mathbb R^2$, but this would contradict the Borsuk–Ulam theorem. Therefore, we deduce \eqref{gc}.
	
	We next assert that
		\begin{equation}\label{wp}
			{\rm tr}(P_{f(X)} P_{f(Y)}) = {\rm tr}(P_X P_Y).
		\end{equation}
	Assume first that $X$ and $Y$ are orthogonal. Lemma \ref{hole} implies that $\hol{X}{Y} = \left\langle S \right]_1$, so it follows from \eqref{gc} that $\hol{f(X)}{f(Y)} = \left\langle f(X)+f(Y) \right]_1$. Another application of Lemma \ref{hole} yields that $f(X)$ and $f(Y)$ are orthogonal, as desired. Suppose now that $X$ and $Y$ are not orthogonal and denote $t = \tfrac{1}{1 + \sqrt{{\rm tr}(P_X P_Y)}} \in \left(\tfrac{1}{2},1\right)$. By Lemma \ref{hole}, $\hol[t]{X}{Y}$ is a singleton. We claim that
	\begin{equation}\label{ghol}
		f (\hol[t]{X}{Y}) = \hol[t]{f(X)}{f(Y)}.
	\end{equation}
	Indeed, the left-hand side is contained in the right-hand side by Lemma \ref{lemma1}. Let now $W \in \hol[t]{f(X)}{f(Y)}$. Then $W \in \left\langle f(X)+f(Y) \right]_1$ by Lemma \ref{hole}, so \eqref{gc} yields that $W = f(W')$ for some $W' \in \left\langle S \right]_1$. Hence,
	\begin{equation}\label{pxyw}
		t(P_{f(X)} + P_{f(Y)}) + (1-2t) P_{f(W')} = P_{W''}
	\end{equation}
	for some $W'' \in \G{1}$. Then we have $W'' \in \hol[\tfrac{t}{2t-1}]{f(X)}{f(Y)}$, hence another application of \eqref{gc} implies that $W'' = f(W''')$ for some $W''' \in \left\langle S \right]_1$.
	Denote
	$$
	A = t(P_X+P_Y) + (1-2t) P_{W'} - P_{W'''}.
	$$
	By \eqref{pxyw}, $L(A) = 0$. We assert that $A = 0$. Indeed, $A = a P_Z + b P_{S \cap Z^\bot}$ for some $Z \in \left\langle S \right]_1$ and $a, b \in \mathbb R$. Since $f|_{\left\langle S \right]_1}$ is injective, $f(Z) \ne f(S \cap Z^\bot)$, so $P_{f(Z)}$ and $P_{f(S \cap Z^\bot)}$ are linearly independent. Now $0 = L(A) = a P_{f(Z)} + b P_{f(S \cap Z^\bot)}$ implies that $a=b=0$ and $A = 0$, which completes the proof of \eqref{ghol}.
	
	By \eqref{ghol}, $\hol[t]{f(X)}{f(Y)}$ is a singleton. Another application of Lemma \ref{hole} yields that $t = \tfrac{1}{1 + \sqrt{{\rm tr}(P_{f(X)} P_{f(Y)})}}$, so \eqref{wp} holds.

We have shown that \eqref{wp} holds whenever $X, Y \in \G 1$ are such that $f(X) \ne f(Y)$. By Lemma \ref{const}, the same holds for any pair from $\left\langle X+Y \right]_1$.

We will next show $f$ is injective. If $\dim H = 2$, there is nothing more to do, so assume that $\dim H \ge 3$. Seeking a contradiction, suppose that there exist pairwise distinct $X, Y, Z \in \G 1$ such that $f(X) \ne f(Y)$ and $f(Z) = f(X)$. Denote $S = X+Y$ and let $Z' = (S+Z) \cap S^\bot$. By Lemma \ref{const}, $Z \not \subset S$, thus $Z' \in \G 1$. Let $Y' \in \left\langle S \right]_1 \setminus \{X\}$ be non-orthogonal to $X$. By the previous paragraph, $f(X)$ and $f(Y')$ are distinct and non-orthogonal. Since $Z'$ is orthogonal $Y'$, $f(Z')$ is either equal or orthogonal to $f(Y')$, so $f(Z') \ne f(X)$. Because $Z'$ is orthogonal to $X$, $f(Z')$ is orthogonal to $f(X)$, which equals $f(Z)$. By the previous paragraph, $Z'$ is orthogonal to $Z$. Hence,
$$
\{0\} = (S+Z) \cap (S+Z)^\bot = Z' \cap Z^\bot = Z',
$$
a contradiction. This contradiction shows that, since $f$ is not constant, it must be injective. Thus, \eqref{wp} holds for all $X, Y \in \G 1$. The conclusion of the theorem now follows from Wigner's theorem, see e.g. \cite{Geher1}.
\end{proof}

\section{Final remarks}
Consider a linear map $L$ on ${\mathcal F}_{s}(H)$
satisfying 
$$
L({\mathcal P}_{k}(H))\subset {\mathcal P}_{k}(H)
$$
for a certain $k \in \mathbb N$, $k < \dim H$. As above, $L$ induces a transformation $f$ of ${\mathcal G}_{k}(H)$ which is not necessarily injective.
The general case can be reduced to the case when $\dim H\ge 2k$. 

For subspaces $M$ and $N$ satisfying $\dim M<k<\dim N$ and $M\subset N$
we denote by $[M,N]_k$ the set of all $X \in \G k$ such that $M\subset X\subset N$.
For any $X,Y\in {\mathcal G}_{k}(H)$ we have 
$$\hol{X}{Y} = \left\{ Z \in \G k \, : \, P_X + P_Y - P_Z \in \PH k \right\} \subset [X\cap Y,X+Y]_{k}$$
and the inverse inclusion holds if and only if $X,Y$ are compatible, i.e.
there is an orthonormal basis of $H$ such that $X$ and $Y$ are spanned by subsets of this basis.
If $X$ and $Y$ are orthogonal, then $\hol{X}{Y}=\left\langle X+Y\right]_{k}$ and 
$$f(\left\langle X+Y\right]_{k})\subset \hol{f(X)}{f(Y)}\subset \left\langle f(X)+f(Y)\right]_{k}.$$
As in the proof of Lemma \ref{contin}, we show that for any $(2k)$-dimen\-sional subspace $S\subset H$
the restriction of $f$ to $\left\langle S\right]_k$ is continuous. 
In the case when $k=1$,  the restriction of $f$ to any projective line is a continuous map to a projective line.

In the general case, a line of ${\mathcal G}_{k}(H)$ is a subset of type $[M,N]_{k}$,
where $M$ is a $(k-1)$-dimensional subspace contained a $(k+1)$-dimensional subspace $N$.
This line can be identified with the line of $\left\langle M^{\perp} \right]_1$ associated to the $2$-dimensional subspace $N\cap M^{\perp}$.
Two distinct $k$-dimensional subspaces are contained in a common line if and only if they are adjacent, i.e. 
their intersection is $(k-1)$-dimensional. 
If $X,Y\in {\mathcal G}_{k}(H)$ are adjacent, then the line containing them is $[X\cap Y, X+Y]_{k}$.
It was noted above that this line coincides with $\hol{X}{Y}$ only in the case when $X$ and $Y$ are compatible.
If $X$ and $Y$ are non-compatible, then $\hol{X}{Y}$ is a  subset of the line $[X\cap Y, X+Y]_{k}$ homeomorphic to a circle. 

For every line there is a $(2k)$-dimensional subspace $S$ such that $\left\langle S\right]_k$ contains this line,
i.e. the restriction of $f$ to each line is continuous. 
Using analogous arguments as in the proof of Lemma \ref{const}, 
we establish that the restriction of $f$ to every line is either injective or constant;
but we are not be able to show that $f$ sends lines to subsets of lines. 

On the other hand, if $f$ is injective, then it is adjacency and orthogonality preserving 
(see \cite{ACh, Geher, Pankov3} for the details).
By \cite{Pankov2}, this immediately  implies that $f$ is induced by a linear or conjugate-linear isometry if $\dim H>2k$
and there is one other option for $f$ if $\dim H=2k$.

\end{document}